\documentclass[a4paper,11pt]{article}

\usepackage[latin1]{inputenc}
\usepackage[english]{babel}
\usepackage{graphicx}
\usepackage{amssymb}
\usepackage{amsmath}%
\usepackage{amsfonts}%
\usepackage{amsthm}
\usepackage{url}

\usepackage{algorithm}
\usepackage{algorithmic}

\newtheorem{theorem}{Theorem}

\newtheorem{corollary}{Corollary}

\newtheorem{definition}{Definition}
\newtheorem{example}{Example}
\newtheorem{lemma}{Lemma}

\newtheorem{problem}{Problem}

\newtheorem{remark}{Remark}
\newtheorem{obs}{Observation}
\numberwithin{equation}{section}
\DeclareMathOperator{\argmin}{argmin}
\DeclareMathOperator{\argmax}{argmax}
\DeclareMathOperator{\rank}{rank}
\DeclareMathOperator{\sign}{sign}

\title{Low-Rank Matrix Approximation in the Infinity Norm} 
\author{Nicolas Gillis\thanks{Department of Mathematics and Operational Research, 
University of Mons, Rue de Houdain 9, 7000 Mons, Belgium. 
 E-mail: nicolas.gillis@umons.ac.be. 
 N.~Gillis acknowledges the support by the F.R.S.-FNRS (incentive grant for scientific research no F.4501.16) and 
by the ERC (starting grant no 679515).} \and 
 Yaroslav Shitov\thanks{National Research University Higher School of Economics, 20 Myasnitskaya Ulitsa,
Moscow 101000, Russia. E-mail: yaroslav-shitov@yandex.ru}}
\date{}

\begin{document}
\maketitle

\begin{abstract}
The low-rank matrix approximation problem with respect 
to the entry-wise $\ell_{\infty}$-norm is the following: given a matrix $M$ and a factorization rank $r$, 
find a matrix $X$ whose rank is at most $r$ and that minimizes $\max_{i,j} |M_{ij} - X_{ij}|$. 
In this paper, we prove that the decision variant of this problem for $r=1$ is NP-complete using a reduction from the problem `not all equal 3SAT'. 
 We also analyze several cases when the problem can be solved in polynomial time, 
and propose a simple practical heuristic algorithm which we apply on the problem of the recovery of a quantized low-rank matrix. 
\end{abstract} 


\section{Introduction}

Low-rank matrix approximations (LRA) are key problems in numerical linear algebra, and have become a central tool in data analysis and machine learning; see, e.g.,~\cite{UHZB14}. 
A possible formulation for LRA is the following: given a matrix $M \in \mathbb{R}^{m \times n}$ and a factorization rank $r$, solve 
\begin{equation}  \label{lowrankapp}
\min_{X \in \Omega} \; ||M - X|| \quad \text{ such that } \quad \rank(X) \leq r, 
\end{equation} 
for some given (pseudo) norm $||.||$. In this paper, we focus on unconstrained variants, that is, $\Omega = \mathbb{R}^{m \times n}$, although there exists many important variants of~\eqref{lowrankapp} that take constraints into account, e.g., 
nonnegative matrix factorization~\cite{LS99}, 
independent component analysis~\cite{C94} and 
sparse principal component analysis~\cite{d2007direct}.  

When the norm $||.||$ is the Frobenius norm, that is, $||M - X||_F^2 = \sum_{i,j} (M_{ij}-X_{ij})^2$, 
the problem can be solved using the singular value decomposition and is closely related to principal component analysis~\cite{golub2012matrix}. 
In practice, it is often required to use other norms, e.g., 
the $\ell_1$ norm which is more robust to outliers~\cite{SWZ17}, 
weighted norms that can be used when data is missing~\cite{GZ79, KBV09}, 
$\sum_j ||M(:,j)-X(:,j)||_2^p$ for $p \geq 1$ which can model different situations depending on the value of $p$~\cite{CW15}, and 
Kullback-Leibler (KL) divergence when Poisson noise is present~\cite{chi2012tensors}. 
However, as soon as the norm is not the Frobenius norm, \eqref{lowrankapp}  
becomes difficult in general; 
in particular it was proved to be NP-hard for all the previously listed cases except for the KL divergence~\cite{GV15c, GG10c, SWZ17}. 
(For the KL divergence, proving NP-hardness is, to the best of our knowledge, an open problem.)  

In this paper, we focus on the variant with the component-wise $\ell_{\infty}$ norm: 
\begin{equation}  \label{lowrankinf}
\min_{X \in \mathbb{R}^{m \times n}} \; ||M - X||_{\infty} \quad \text{ such that } \quad \rank(X) \leq r, 
\end{equation} 
where $||M - X||_{\infty}  = \max_{i,j} |M_{ij} - X_{ij}|$. We will refer to this problem as $\ell_{\infty}$ LRA. It should be used when the noise added to the low-rank matrix follows an \emph{i.i.d.\@ uniform distribution}. We will also use the notation $\ell_p$ LRA for the LRA problem where the norm used is the component-wise $\ell_p$ norm. 

\subsection{Previous results and applications}

When $m=n$ and $r=\min(m,n)-1$, \eqref{lowrankinf} was studied in~\cite{PR93} and corresponds to the problem of distance to robust nonsingularity which can be stated as follows: what is the rank-deficient matrix $X$ that is the closest to $M$ in the component-wise infinity norm? In this particular case, \eqref{lowrankinf} was shown to be NP-hard~\cite{PR93}. 
In~\cite{GT01, GT11}, Goreinov and Tyrtyshnikov  obtained an approximate solution to $\ell_{\infty}$ LRA using the so-called rank-$r$ skeleton approximation which uses $r$ columns and $r$ rows of the given matrix.  
Juditsky et al.~\cite{juditsky2011low} linked a variant of~\eqref{lowrankinf} (where the row range of $X$ is constrained to be contained in the row range of $M$) with the synthesis problem of compressed sensing and provided randomized algorithms with optimality guarantees.  
Very recently, Chierichetti et al.~\cite{CGKe17} proposed 
provably good approximation algorithms for $\ell_p$ LRA for any $p \geq 1$ using a subset of the columns of the input matrix $M$ to span the columns of $X$. 
An application of~\eqref{lowrankinf} is the recovery of a low-rank matrix from a quantization~\cite{LJ17} (roughly speaking, its rounding to some precision); see, e.g.,~\cite{gersho1992vector} for more details on quantization. 
For example, assume we are given a real rank-$r$ matrix where each entry has been rounded to the nearest integer. Given such a matrix $M$ (which in general will have full rank), the problem is to find a rank-$r$ $X$ such that $||M-X||_{\infty} \leq 0.5$ (note that rounding can be assimilated to noise distributed uniformly in the interval [-0.5,0.5]); see Section~\ref{algoappl} for some examples.  

As far as we know, there is not as much literature on $\ell_{\infty}$ LRA~\eqref{lowrankinf} compared to other variants. 
A plausible explanation is that, in practice, and especially in data analysis applications, 
using the $\ell_{\infty}$ norm is not very useful and is in particular extremely sensitive to outliers. 
Moreover, even if $M \neq 0$, the zero matrix could be an optimal solution of~\eqref{lowrankinf} (which is not possible for any other $\ell_{p}$-norm). 
\begin{example} \label{ex1} 
Let 
\[
M = \left( \begin{array}{cc} 
1 & 1 \\ 1 & -1 
\end{array} \right). 
\]
We have 
\[ 
\min_{\rank(X) \leq 1} ||M - X||_{\infty} 
\;  = \;   
\min_{u \in \mathbb{R}^2,v \in \mathbb{R}^2} ||M-uv^T||_{\infty} \; = \;  ||M||_{\infty} \;  = \;  1. 
\]
In fact, assume the minimum is strictly smaller than 1. This requires $u_1 v_i > 0$ $(i=1,2$), $u_2 v_1 > 0$ and $u_2 v_2 < 0$ which is not possible. 
\end{example}

However, there are several applications as mentioned above; namely, distance to singularity, compressed sensing, and recovery of a quantized low-rank matrix. 
Except for the case $m=n$ and $r = n-1$ which was proved to be NP-hard by Poljak and Rohn~\cite{PR93}, there is, to the best of our knowledge, not a very good understanding of the computational complexity of~\eqref{lowrankinf}. 
The main goal of this paper is to shed some light on this question; in particular proving NP-completeness of~\eqref{lowrankinf} when $r=1$ (Theorem~\ref{thrnpcom}).

\subsection{Outline of the paper and contributions}

In this paper we mainly focus on rank-one $\ell_{\infty}$ LRA, that is, \eqref{lowrankinf} with $r=1$. 
In Section~\ref{pblinNP}, we show that the decision version of rank-one $\ell_{\infty}$ LRA is in NP. 
In fact, we show that if the sign pattern of $X$ is known, then the decision version of rank-one $\ell_{\infty}$ LRA can be solved in polynomial time by finding a solution to a system of linear inequalities. 
This is an important result because it turns out that, in most cases, the number of possible sign patterns of $X$ can be reduced drastically;  e.g., $X$ can be assumed to be nonnegative if $M$ is. 
In Section~\ref{NPcompl}, we prove that rank-one $\ell_{\infty}$ LRA is NP-complete using a reduction from `not all equal 3SAT', 
with an intermediate problem on directed graphs (namely, the problem of making, if possible, a directed graph acyclic by reversing the direction of a particular subset of the edges). 
Finally, in Section~\ref{algoappl}, we propose a simple heuristic algorithm and apply it on the recovery of quantized low-rank matrices.

\section{Decision version of rank-one $\ell_{\infty}$ LRA}  \label{pblinNP}

Let us define formally the decision version of rank-one $\ell_{\infty}$ LRA:
  
  \begin{problem} \label{probDellinf} (D-$\ell_\infty$-R1A($M$,$k$)) 
  
\noindent Given: A real $m$-by-$n$ matrix $M$ and a real number $k \geq 0$. 

\noindent Question:  Does there exist $u \in \mathbb{R}^{m}$ and $v  \in \mathbb{R}^{n}$ such that $\max_{i,j} |M_{ij} - u_i v_j| \leq k$? 
If yes, output a solution $(u,v)$.  
\end{problem}

\begin{lemma} \label{lemsignpat}  
If the sign pattern of $u$ or $v$ is given as a part of the data, 
then D-$\ell_\infty$-R1A($M$,$k$) can be solved in polynomial time in $m$ and $n$, 
namely in $O\big( mn (m+n)^3 \log mn \big)$ arithmetic operations.  
\end{lemma}
\begin{proof}
First, we can assume without loss of generality (w.l.o.g.) that for each row (resp.\@ column) of $M$, there is at least one entry whose absolute value is larger than $k$, that is, for all $i$ (resp.\@ $j$), there exists $l$ (resp.\@ $p$) such that $|M_{il}| > k$ (resp.\@ $|M_{pj}| > k$). 
In fact, if it is not the case, then one can trivially choose $u_i = 0$ (resp.\@ $v_j = 0$) in a solution of D-$\ell_\infty$-R1A($M$,$k$) and reduce the problem to a submatrix of $M$. 
This implies that we can assume w.l.o.g.\@ that $u \neq 0$ and $v \neq 0$ since $|M_{ij}| > k \Rightarrow u_i v_j \neq 0$. 

Second, if we assume that the sign pattern of $u$ is known we can assume w.l.o.g.\@ that $u > 0$. 
In fact, if some entries of $u$ are negative, we can flip their signs along with the signs of the entries of the corresponding rows of $M$ to obtain an equivalent problem. 

Therefore, if the sign pattern of $u$ is known, D-$\ell_\infty$-R1A($M$,$k$) can be reduced to finding $u > 0$ and $v$ such that 
\[
-k \leq |M_{ij} - u_i v_j| \leq k 
\quad \iff \quad 
M_{ij}-k \leq u_i v_j \leq M_{ij}+k. 
\]
Moreover, by the scaling degree of freedom (that is, $uv^T = (\alpha u)(\alpha^{-1} v)^T$ for any $\alpha \neq 0$), we can assume w.l.o.g.\@ that $u \leq 1$. 
Defining $s_i = (u_i)^{-1} \geq 1$, 
this problem is equivalent to finding $s \geq 1$ and $v$ such that 
\[
s_i \, (M_{ij}-k) 
\; \leq \; 
 v_j \; \leq \; 
s_i \, (M_{ij}+k) \quad \text{ for all } i,j. 
\] 
This is a system of $2mn+m$ linear inequalities with $m+n$ variables, 
where each inequality contains at most two variables, 
and can be solved in $O\big( I N^3 \log I \big) = O\big( mn (m+n)^3 \log mn \big)$ arithmetic operations where $I=2mn+m$ is the number of inequalities and $N=m+n$ the number of variables~\cite{megiddo1983towards}. 

If the sign pattern of $v$ is known, by symmetry ($||M-uv^T||_{\infty} = ||M^T - v u^T||_{\infty}$), the same result holds.  This completes the proof. 
\end{proof}

\begin{corollary} \label{lemMnonneg}  
If $M \geq 0$,  D-$\ell_\infty$-R1A($M$,$k$) can be solved in polynomial time. 
\end{corollary} 
\begin{proof}
In fact, in that case, one can assume w.l.o.g.\@ that $u \geq 0$ and $v \geq 0$ so that the result follows from Lemma~\ref{lemsignpat}.  
\end{proof}

\begin{remark} 
It is interesting to note that some rank-one LRA problems are NP-hard even when $M \geq 0$; e.g., for the $\ell_1$ norm~\cite{GV15c} and weighted norms~\cite{GG10c}. 
\end{remark}

Lemma~\ref{lemsignpat} implies that D-$\ell_\infty$-R1A($M$,$k$)  can be solved in $O\big( 2^{\min(m,n)} mn (m+n)^3 \log mn \big)$ operations since one can try all the possible sign patterns for $u$ (or $v$ if $n \leq m$). 
It is possible to achieve a better complexity result by identifying the connected component of a particular bipartite graph. 
\begin{definition}
Given $M$ and $k$, $G_b(M,k) = (V_1 \cup V_2, E)$ is the bipartite graph with 
$V_1 = \{v_1,v_2,\dots,v_m\}$, 
$V_2 = \{v_1',v_2',\dots,v_n'\}$, 
and 
$(v_i,v_j') \in E \iff |M_{ij}| > k$.  
\end{definition}

\begin{lemma} \label{lemconcomp} 
For D-$\ell_\infty$-R1A($M$,$k$), the number of possible sign patterns in a solution $u$ can be reduced to $2^{d-1}$ where $d$ is the number of connected components of $G_b(M,k)$ discarding the isolated vertices.  
\end{lemma}
\begin{proof} 
Note that the isolated vertices of $G_b(M,k)$ correspond to rows and columns of $M$ whose entries have absolute value smaller than $k$ and for which one can set w.l.o.g.\@ the corresponding entry of $u$ or $v$ to zero; see the proof of Lemma~\ref{lemsignpat}.  

The rest of the proof follows from the fact that any solution $(u,v)$ of D-$\ell_\infty$-R1A($M$,$k$) must satisfy the following:  
\[
|M_{ij}| > k \quad \Rightarrow \quad \sign(u_i v_j) = \sign(M_{ij}), 
\] 
where 
\[
\sign(x) = \left\{ 
\begin{array}{cc}
1 & \text{ if $x > 0$,} \\
0 & \text{ if $x = 0$,} \\
-1 & \text{ if $x < 0$.} 
\end{array}
\right. 
\]
Therefore, fixing the sign of an entry of $u$ imposes the sign for all the entries of $u$ and $v$ contained in the same connected component. This makes $2^{d}$ possible sign patterns for $u$ and $v$.  
It can be reduced to $2^{d-1}$ since $uv^T = (-u)(-v)^T$ hence half of the possible sign patterns can be discarded 
(this can be achieved for example by imposing arbitrarily the sign of one entry of $u$ or $v$).  
\end{proof}

For any submatrix of $M$ corresponding to a connected component of $G_b(M,k)$, there must exist a completion with $\pm 1$ of the entries smaller than $k$ such that its sign pattern has rank one, otherwise the answer to D-$\ell_\infty$-R1A($M$,$k$) is NO; 
see, e.g., Example~\ref{ex1} for any $k < 1$. 
This observation can be used to quickly obtain a lower bound on $k$ in order for the answer to D-$\ell_\infty$-R1A($M$,$k$) to possibly be YES 
(hence also a lower bound for rank-one $\ell_{\infty}$ LRA). 

\begin{theorem} \label{theoremNP} 
D-$\ell_\infty$-R1A($M$,$k$) is in NP, 
and can be solved in $O\big( 2^{d} mn (m+n)^3 \log mn \big)$ 
operations where $d$ is the number of connected components of $G_b(M,k)$ discarding isolated vertices.   
\end{theorem}
\begin{proof}
This follows directly from Lemmas~\ref{lemsignpat} and \ref{lemconcomp}. 
\end{proof}
Theorem~\ref{theoremNP} implies that for D-$\ell_\infty$-LRA to be a difficult problem, the number of connected components has to be high. 
This will motivate our construction in our NP-completeness proof where we will use a square matrix for which only the diagonal entries are larger than $k$ so that the number of connected components is maximal, namely $d=m=n$.

Theorem~\ref{theoremNP} also implies that D-$\ell_\infty$-R1A($M$,$k$) can be solved in polynomial time 
if the number of connected components $d$ satisfies $d = O\big(\log(\min(m,n)) \big)$.

\section{NP-completeness of D-$\ell_\infty$-R1A($M$,$k$)} \label{NPcompl}

The goal of this section is to prove that D-$\ell_\infty$-R1A is NP-hard. In order to do this, we construct a polynomial time reduction from the problem known as NOT-ALL-EQUAL 3-SAT. Recall that a \textit{literal} associated with a set $X$ of Boolean variables is either an element of $X$ or a negation of it. 

\begin{problem}\label{prob12} (NOT-ALL-EQUAL $3$-SAT) 

\noindent Given: A set $X$ of variables and a set $L$ of $3$-tuples of literals.

\noindent Question: Does there exist an assignment of the variables in $X$ to $\{0,1\}$ for which every tuple in $L$ has at least one false literal and at least one true literal? 
\end{problem} 

Since NOT-ALL-EQUAL $3$-SAT is NP-complete (see~\cite{Karp}), constructing a polynomial time reduction from it to D-$ell_\infty$-R1A would mean the NP-hardness of the latter problem. In order to present such a reduction, we need to recall the definitions of some basic concepts in graph theory. An \textit{oriented graph} $G$ is defined as a pair of sets $V$ and $E\subset V^2$. The elements of $V$ are called \textit{vertices}, a pair $(a,b)\in E$ is an \textit{edge passing from} $a$ \textit{to} $b$, and vertices $a,b$ are \textit{adjacent} if there is an edge passing between them. We assume that $V=\{1,\ldots,n\}$ and that at most one of the pairs $(a,b)$, $(b,a)$ belongs to $E$ for all $a,b \in V$. 
A sequence $(a_0,\ldots,a_k)$ of vertices is called a \textit{cycle} if $a_0=a_k$ and there is an edge passing from $a_{i-1}$ to $a_i$ for all $i$. 
A \textit{two-coloring} of $G$ is a partition of $V$ into the union of two disjoint sets $W$ and $B$. A subset $U\subset V$ is called \textit{monochromatic} with respect to $(W,B)$ if either $U\subset W$ or $U\subset B$. Let us introduce the auxiliary problem which we use as a tool in our NP-hardness proof.

\begin{problem} \label{probgraph} \hspace{0.1cm} \,  

\noindent Given: An oriented graph $G=(V,E)$ and a set $D$ of pairs of non-adjacent vertices.

\noindent Question: Is there a two-coloring $(W,B)$ of $V$ such that 
\begin{itemize}
\item[(i)]  no pair in $D$ is monochromatic, and 
\item[(ii)] the graph obtained from $G$ by reversing the edges passing between $W$ and $B$ has no cycle. 
\end{itemize}
\end{problem}

\begin{lemma}\label{lemgraph}
Problem~\ref{probgraph} is NP-complete.
\end{lemma}

\begin{proof}
Let us construct the graph $G$ depending on an instance $(X,L)$ of Problem~\ref{prob12}: 
\begin{itemize} \setlength{\itemindent}{0.5cm}  
\item[Step 1.] For every variable $x\in X$, we create two vertices corresponding to $x$ and the negation of $x$, and we add to $D$ the pair containing these two. 

\item[Step 2.] For every tuple $(y_1,y_2,y_3)\in L$, we create three new vertices corresponding to $y_1$, $y_2$ and $y_3$, and we draw a cycle on these vertices.  

\item[Step 3.] We add a pair in $D$ containing a vertex created in Step~1 with a vertex created in Step~2 if they correspond to literals that are negations of each other.  
\end{itemize} 

Clearly, this graph can be constructed in polynomial time, and a two-coloring of $V$ leaves no pair in $D$ monochromatic if and only if it assigns the same color for every occurrence of a literal $y$ and the other color for the negation of $y$. 
If $y_1$, $y_2$ and $y_3$ have all the same color and $(y_1,y_2,y_3)\in L$, then the item (ii) in Problem~\ref{probgraph} does not require changes in the edge directions between $y_1$, $y_2$ and $y_3$, so these vertices remain a cycle. This implies that any acceptable two-coloring of $G$ for Problem~\ref{probgraph} will not have $y_1$, $y_2$ and $y_3$ of the same color hence will correspond to a valid assignment for $(X,L)$. 
On the other hand, any valid assignment of $(X,L)$ corresponds to a coloring in which $y_1$, $y_2$ and $y_3$ have different colors 
for all $(y_1,y_2,y_3)\in L$. In other words, two of the three edges passing between vertices in $y_1$, $y_2$ and $y_3$ will change their directions as in item~(ii) in Problem~\ref{probgraph}, which means that the resulting graph will possess no cycle hence any valid assignment of $(X,L)$ corresponds to an acceptable two-coloring of $G$. 
\end{proof}

We are now ready to present a reduction from Problem~\ref{probgraph} to D-$ell_\infty$-R1A. 
\begin{definition}
Let $G=(V,E)$ and $D$ be defined as in the formulation of Problem~\ref{probgraph}. 
We define the matrix $\mathcal{M}=\mathcal{M}(G,D) \in \{-1,0,1\}^{|V| \times |V|}$ with rows and columns indexed by elements of $V=\{1,2,\dots,n\}$ as follows:
\begin{itemize} \setlength{\itemindent}{0cm}   
\item[(1)] $\mathcal{M}_{ii}=2$ for all $i$,

\item[(2)] $\mathcal{M}_{ij}=\mathcal{M}_{ji}=-1$ if $\{i,j\}\in D$,

\item[(3)] $\mathcal{M}_{ij}=-1$, $\mathcal{M}_{ji}=1$ if $(i,j)\in E$,

\item[(4)]  $\mathcal{M}_{ij}=0$ otherwise.
\end{itemize}
\end{definition}

Before we can prove that $\mathcal{M}$ leads to a reduction, 
we need a result stating that any partial order relation on a set can be extended to a total order relation. In terms of graphs, this result can be stated as follows.

\begin{obs}\label{obserpermut}
Let $G=(V,E)$ be an oriented graph without cycles. Then there exists a total ordering $\succ$ of $V$ such that $(a,b)\in E$ implies $a\succ b$.
\end{obs}

\begin{theorem}\label{thrred}
The pair $(G,D)$ is a yes-instance of Problem~\ref{probgraph} if and only of there are real numbers 
$\{u_i\}_{i=1}^{|V|}$ and $\{v_j\}_{i=1}^{|V|}$ such that $|\mathcal{M}_{ij}-u_iv_j| \leq 3/2-0.001 |V|^{-6}$ for all $i,j$. 
\end{theorem}

\begin{proof}
Assume $(G,D)$ admits a valid coloring $(W,B)$ as in Problem~\ref{probgraph}, 
and let $G'$ be the directed graph obtained after the transformation in item~(ii) of Problem~\ref{probgraph}. 
Consider the matrix $N$ defined by $N_{ij}=\mathcal{M}_{ij}$ if $\{i,j\}$ is monochromatic and $N_{ij}=-\mathcal{M}_{ij}$ otherwise. 
The matrix $N$ has $2$'s on the diagonal, $-1$'s at every position $(i,j)$ which is an edge of $G'$, and zeros and ones everywhere else. 
Since $G'$ has no cycle, by Observation~\ref{obserpermut}, there is a permutation matrix $C$ (corresponding to an ordering) such that all the $-1$'s are located below the main diagonal of the matrix $C^{-1} N C$. Since the permutations of rows and columns and multiplications of them by $-1$ do not change our ability or inability to approximate a matrix, it is sufficient to prove the existence of real numbers $\{u_i\}_{i=1}^{|V|}$ and $\{v_j\}_{i=1}^{|V|}$ such that $|M_{ij}-u_iv_j| \leq 3/2-0.001 |V|^{-6}$ 
for any $n\times n$ matrix $M$ with the $2$'s on the diagonal, $-1$'s at some positions below the diagonal, and zeros and ones everywhere else. Towards this end, one can check that it suffices to set 
$$
u_i=\frac{1}{\sqrt{2}}-i\varepsilon,\,\,\,
v_j=\frac{1}{\sqrt{2}}+j\varepsilon+\varepsilon^{1.5},$$
where $\varepsilon=0.1 |V|^{-4}$. 

Conversely, assume that the numbers $\{u_i\}_{i=1}^{|V|}$ and $\{v_j\}_{i=1}^{|V|}$  are such that 
$|\mathcal{M}_{ij}-u_iv_j| \leq 3/2-0.001 |V|^{-6} < 3/2$ for all $i,j$.  
Since $\mathcal{M}_{ii}=2$ for all $i$, we have $|u_iv_i-2|<3/2$ 
which implies that $u_i$ and $v_i$ are non-zero and have the same sign. A relabeling of indices does not change the properties we discuss, so we can assume $|v_1|\leq \ldots\leq |v_n|$. We define $W$ (resp.\@ $B$) as the set of all $i$ such that $u_i>0$ (resp.\@ $u_i<0$), and we are going to prove that $(W,B)$ is a valid coloring of $G$ as in Problem~\ref{probgraph}. 
We define the matrix $N$ by multiplying the rows and columns of $\mathcal{M}$ with indices in $B$ by $-1$, and define $u'$ and $v'$ by multiplying the entries in $u$ and $v$, respectively, with indices in $B$ by $-1$. 
We have $\left|N_{ij}-u'_iv'_j\right|<3/2$ and $u' \geq 0$ and $v' \geq 0$. 
For  $j>i$, $v'_j\geq v'_i$ which implies $u'_iv'_j \geq u'_i v'_i > 1/2$, so $N_{ij}\neq-1$ if $j>i$. In other words, all the $-1$'s of $N$ are located below the main diagonal, and we have $N_{ab}=N_{ba}=1$ for all the positions $(a,b)$ in $D$. This means that the item~(i) of Problem~\ref{probgraph} is satisfied, and the graph as in item~(ii) has indeed no cycle because the edges $(a,b)$ of this graph correspond to the positions of the $-1$'s in $N$ which are all located below the main diagonal. 
\end{proof} 

Now we can determine the complexity status of D-$ell_\infty$-R1A by proving that it is NP-complete.

\begin{theorem}\label{thrnpcom}
The D-$ell_\infty$-R1A problem is NP-complete.
\end{theorem}

\begin{proof}
The function $(G,D)\to(\mathcal{M},3/2-0.001|V|^{-6})$ can be computed in polynomial time, and Theorem~\ref{thrred} proves that it is a reduction from Problem~\ref{probgraph} to D-$ell_\infty$-R1A. Therefore, D-$ell_\infty$-R1A is NP-hard by Lemma~\ref{lemgraph}. The membership of D-$ell_\infty$-R1A in NP is stated in Theorem~\ref{theoremNP}. 
\end{proof}

%

\begin{example}
For the matrix 
\[
M = \left( \begin{array}{ccccc}
2 &  0 &  1 &  1 & -1 \\ 
-1 &  2 & -1 & -1 &  0 \\  
-1 &  1 &  2 & -1 & -1 \\ 
-1 &  1 &  1 &  2 & -1 \\ 
1 & -1 &  0 &  1 &  2 \\  
			\end{array} \right), 
		\]
	there exist a permutation and a sign flip of the rows and columns so that the negative entries are below the main diagonal: 
	 $N_{ij} = M_{\pi_i \pi_j} s_{\pi_i} s_{\pi_j}$ with	$s = (1,-1,1,1,1)$ and $\pi = (5,2,1,4,3)$, with 
	\[
	N = 
\left( \begin{array}{ccccc} 
  2 &   1 &   1 &   1 &   0 \\ 
  0 &   2 &   1 &   1 &   1 \\ 
 -1 &   0 &   2 &   1 &   1 \\ 
 -1 &  -1 &  -1 &   2 &   1 \\ 
 -1 &  -1 &  -1 &  -1 &   2 \\ 
\end{array} \right). 
		\] 
		We have $\min_{u,v} ||M-uv^T||_{\infty} = 1.3456 < 3/2$. 
For the matrix 
\[
M = \left( \begin{array}{ccccc}
 2  &   1 &   1&    -1 &    1 \\
    -1   &  2  &  -1   & -1 &    0  \\
    -1   &  0   &  2   &  1  &   1  \\
     0   &  1   & -1    & 2  &  -1  \\
    -1    &-1   & -1    & 1 &    2  \\ 
		\end{array} \right)
		\]
		there is no such sign flip and permutation, and $\min_{u,v} ||M-uv^T||_{\infty} = 3/2$.  
\end{example}

\section{Heuristic algorithm and application to the recovery of quantized low-rank matrices}  \label{algoappl}

This goal of this section is to describe a simple heuristic algorithm for $\ell_{\infty}$ LRA and apply it for the recovery of quantized low-rank matrices. It will allow us to get some more insight on this problem. 
The algorithm is available from \url{https://sites.google.com/site/nicolasgillis/} and allows the interested reader to tackle  $\ell_{\infty}$ LRA (in particular the examples presented in this paper can be run directly).  
In this section, all tests are performed using Matlab on a laptop Intel dual CORE i5-3210M CPU @2.5GHz 6Go RAM.

\subsection{Block coordinate descent method}  

A popular approach in optimization is block coordinate descent (BCD): fix a subset of the variables and optimize over the other variables; 
see~\cite{wright2015coordinate} for a recent survey. An crucial aspect of BCD is to make the subproblem easy (and fast) to solve. 
For $\ell_{\infty}$ LRA, a judicious choice is to optimize alternatively over the columns of $U$ and the rows of $V$; 
see Algorithm~\ref{altopt}. 
In fact, the corresponding subproblems are convex and separable, that is, each entry in a column of $U$ (resp.\@ in a row of $V$) can be optimized independently of the other entries in the same column (resp.\@ row); this is described in the next section.  
\begin{algorithm}
\caption{$(U,V)$ = BCD $\ell_{\infty}$ LRA $(M,U_0,V_0)$} 
\label{altopt}
\begin{algorithmic}[1]
\STATE INPUT:  $M \in \mathbb{R}^{m \times n}$, $U_0 \in \mathbb{R}^{m \times r}$ and $V_0 \in \mathbb{R}^{r \times n}$. 
\STATE OUTPUT: $U \in \mathbb{R}^{m \times r}$, $V \in \mathbb{R}^{r \times n}$ so that $||M-UV||_{\infty}$ is minimized. 

   \STATE $U = U_0$, $V = V_0$. 
	\FOR{iter $=1,2,\dots$}
			\STATE $R = M - UV$.
				\FOR{$p=1,2,\dots,r$}
				  \STATE $R = R + U(:,p) V(p,:)$. 
					
					\STATE For all $i$, update $U(i,p) = \argmin_{x} \max_{ \{ j | V(p,j) \neq 0 \} } |R(i,j) - x  \, V(p,j)|$.
					
					\STATE For all $j$, update $V(p,j) = \argmin_{y} \max_{ \{ i | U(i,p) \neq 0 \} } |R(i,j) - U(i,p) \, y|$.
					
					\STATE $R = R - U(:,p) V(p,:)$.
					
					\ENDFOR
	\ENDFOR 
\end{algorithmic}
\end{algorithm}

To initialize Algorithm~\ref{altopt}, we use the optimal solution of $\ell_{2}$ LRA. 
It would be an interesting direction of research to use more sophisticated initialization strategies such as the approximation algorithm proposed in~\cite{CGKe17} that can be refined by Algorithm~\ref{altopt}.


\subsection{Secant method for the subproblem} 

Let us focus on the rank-one subproblem in $v$ (by symmetry, the subproblem in $u$ can be solved in the same way). 
It can be decoupled into $n$ problems in one variable: for $1 \leq j \leq n$, 
we need to solve  
\begin{equation} \label{subpbl1}
 \min_{v_j} \; \max_{i} |M_{ij} - u_i v_j|. 
\end{equation}
The optimal solution is not necessarily unique. Non-uniqueness may happen when $u_i = 0$ and $|M_{ij}|$ is large for some $i$, while uniqueness is guaranteed if $u_i \neq 0$ for all $i$ because the objective function is piece-wise linear with nonzero slopes. 
To make the problem well posed, it makes sense to consider 
\begin{equation} \label{subpbl}
 \min_{v_j} \max_{ \{ i | u_i \neq 0 \} } |M_{ij} - u_i v_j|, 
\end{equation}
with a unique solution which is also optimal for~\eqref{subpbl1}. (Note that if $u = 0$, any $v_j$ is optimal.)   
Let us focus w.l.o.g.\@ on the case $u \geq 0$ by flipping the signs of the rows of $M$ accordingly. 
The global minima is the intersection of two linear functions of the form $f^i_{\pm}(v_j) = \pm(M_{ij} - u_i v_j)$ 
($1 \leq i \leq m$), one with negative slope and one with positive slope. 
Therefore, the optimal solution $v_j^*$ of~\eqref{subpbl} satisfies 
\[ 
v_j^* \in \left\{  \frac{M_{i_1j} + M_{i_2j}}{u_{i_1} + u_{i_2}}  \ \big| \ 1 \leq i_1 \neq i_2 \leq m, u_{i_1} > 0,  u_{i_2} > 0 \right\} . 
\]
Hence solving~\eqref{subpbl} can be done by identifying the two indices $i_1$ and $i_2$ corresponding to the optimal solution. This could be performed by inspection since there are $\frac{m(m-1)}{2}$ possible pairs. 
A more efficient approach is described in the following. Let us define 
\[
i_l = \argmin_{\{ i | u_i \neq 0 \}} \frac{M_{ij}}{u_i} 
\quad 
\text{ and } 
\quad 
i_u = \argmax_{\{ i | u_i \neq 0 \}} \frac{M_{ij}}{u_i}.  
\]
We have $v_l = \frac{M_{i_l j}}{u_{i_l}}  \leq v_j^* \leq   \frac{M_{i_u j}}{u_{i_u}} = v_u$. Since the objective function is convex, 
it is rather straightforward to implement the following \textit{secant method}: 
\begin{enumerate}

\item Initialize $(i_1,i_2) = (i_l, i_u)$. 

\item Intersect the two lines corresponding to the indices $i_1$ (with negative slope) and $i_2$ (with positive slope) to obtain the point $v_c = \frac{M_{i_1j} + M_{i_2j}}{u_{i_1} + u_{i_2}}$.

\item Compute the objective function of~\eqref{subpbl} in $v_j=v_c$ and identify the index $i_a$ that is active, that is, 
the index $i_a$ such that $u_{i_a}\neq 0$ and $|M_{i_a j} - u_{i_a} v_j| = \max_{ \{ i | u_i \neq 0 \} } |M_{ij} - u_i v_j|$. 
If the slope in $i_a$ is positive, replace $i_2$ by $i_a$; otherwise replace $i_1$ by $i_a$. 
If the two indices $(i_1, i_2)$ are active together, the algorithm has converged: return $v_j^* = \frac{M_{i_1j} + M_{i_2j}}{u_{i_1} + u_{i_2}}$; otherwise, return to 2. 

\end{enumerate} 

It turns out that this secant method performs surprisingly well in the sense that it needs a very small number of iterations to terminate. Let us illustrate this on randomly generated instances. 
   
\begin{example}[Numerical experiment on the secant method for~\eqref{subpbl1}] \label{ex2}
We have run the above secant method to solve  $10^4$ problems of the form~\eqref{subpbl1} for different values of $m$, generating each entry of $M$ and $u$ using the normal distribution $N(0,1)$. 
Table~\ref{tableGauss} reports the distribution of the number of iterations needed to solve~\eqref{subpbl1} for the $10^4$ problems, 
along with the total computational time to solve them.  
\begin{center}
 \begin{table}[ht!]
 \begin{center}
\caption{ 
Repartition of the number of iterations performed by the secant method to solve~\eqref{subpbl1} among $10^4$ instances and for different values of $m$, generating each entry of $M$ and $u$ using the normal distribution $N(0,1)$. 
\label{tableGauss}
}
 \begin{tabular}{|c||c|c|c|c|c|c|c|c|c||c|}
 \hline 
m / \# it. &  1   & 2 & 3 &   4 & 5 & 6 & 7 & 8 & 9 &  Time (s.)  \\
 \hline
 10   & 2019  & 1594  & 3277  & 2742  & 363  &  5  &  0  &  0  &  0  &  5.1  \\
  $10^2$   & 199  & 204  & 487  & 4237  & 4099  & 747  &  27  &  0  &  0  &  6.55  \\
   $10^3$   &  20  &  19  &  50  & 2668  & 5200  & 1885  & 154  &  4  &  0  &  8.13  \\
    $10^4$   &  2  &  1  &  11  & 1739  & 5118  & 2769  & 351  &  9  &  0  &  22.03  \\
   $10^5$   &  0  &  0  &  0  & 1258  & 4785  & 3403  & 525  &  28  &  1  &  108.11  \\
  $10^6$ &  0  &  0  &  0  & 916  & 4618  & 3714  & 712  &  38  &  2  &  1685.40  \\ 
  $10^7$   &  0  &  0  &  0  & 687  & 4192  & 4115  & 945  &  60  &  1  &  16793.87  \\ \hline 
	\end{tabular}
 \end{center}
 \end{table}
 \end{center} 
We observe that the secant method requires in average 5 iterations to terminate while 9 are necessary in the worst case.  
\end{example}

\paragraph{Computational cost} Each iteration of the secant method requires $O(m)$ operations hence solving~\eqref{subpbl} requires $O(mK)$ operations where $K$ is the number of iterations performed by the secant method. In practice, as illustrates by Table~\ref{tableGauss}, $K$ will be rather small 
(in all cases we have generated randomly, $K$ is smaller than 9). 
In the worst case, the secant could potentially have to go through all indices with $K = O(m)$ and require $O(m^2n)$ operations. Note that sorting the values $\frac{M_{ij}}{u_i}$ ($1 \leq i \leq m$) in $O(m\log(m))$ operations would allow to perform a bisection in $O(m \log(m))$ operations (getting rid of about $m/2$ indices per iteration). Finally, the total computational cost of Algorithm~\ref{altopt} is $O(mnrK)$ per iteration.


\begin{remark}[$\ell_{\infty}$ LRA with nonnegativity constraints]
Algorithm~\ref{altopt} can easily be adapted to incorporate nonnegativity constraints on $U$ and $V$. In fact, the optimal solution of~\eqref{subpbl1} with the constraint $v_j \geq 0$ is $\max(v_j^*,0)$ where $v_j^*$ is the optimal solution of the unconstrained problem~\eqref{subpbl1}. 
\end{remark}

\subsection{Recovery of quantized low-rank matrices} 


We apply in this section Algorithm~\ref{altopt} to recover a quantized low-rank matrix. First, let us start with a toy example. 
\begin{example} \label{ex3}
Let us generate a simple example with $m = 8$, $n = 5$ and $r = 3$ where $M = UV$ with each entry of $U$ and $V$ generated using the normal distribution $N(0,1)$. We obtain (with two digits of accuracy)
\[
M = 
\left( \begin{array}{cccccccc} 
 0.35 &  0.65 &  0.15 &  0.54 &  1.49 \\ 
 1.17 &  -0.90 &  -1.50 &  -0.52 &  -0.44 \\ 
 1.03 &  -1.12 &  -3.48 &  -1.41 &  -0.17 \\ 
 4.46 &  -1.81 &  3.58 &  2.24 &  -2.23 \\ 
 -1.53 &  -0.60 &  -2.73 &  -1.94 &  -1.10 \\ 
 -2.53 &  2.79 &  1.02 &  0.75 &  3.59 \\ 
 3.38 &  -0.90 &  -1.16 &  0.53 &  0.86 \\ 
 -0.66 &  0.42 &  0.71 &  0.20 &  0.14 \\ 
\end{array} \right)  
\]
and a quantization (here we simply use its rounding to the nearest integer)   
\[
M_q = 
\left( \begin{array}{cccccccc} 
  0 &   1 &   0 &   1 &   1 \\ 
  1 &  -1 &  -1 &  -1 &  0 \\ 
  1 &  -1 &  -3 &  -1 &  0 \\ 
  4 &  -2 &   4 &   2 &  -2 \\ 
 -2 &  -1 &  -3 &  -2 &  -1 \\ 
 -3 &   3 &   1 &   1 &   4 \\ 
  3 &  -1 &  -1 &   1 &   1 \\ 
 -1 &   0 &   1 &   0 &   0 \\ 
\end{array} \right) , 
\]
which has rank 5, and with $||M - M_q||_{\infty} = 0.498$ (entry on second row, third column).  
The optimal solution $X_{\text{svd}}$ of rank-3 $\ell_2$ LRA gives $||M_q-X_{\text{svd}}||_{\infty} = 0.57$ 
while Algorithm~\ref{altopt} provides a solution $X^*$ with $||M_q-X^*||_{\infty} = 0.39$. Note that, for this problem, 
the set 
\[
\{ X | \rank(X) = 3, ||M-X||_\infty \leq 1/2\}
\] 
is rather large and does not only contain a small neighborhood around the matrix $M$. 
It would be an interesting direction for further research to identify conditions for this problem to be well posed 
(e.g., sufficiently many entries of $M-M_q$ close to $\pm 1/2$, or using additional constraints on $X$). 
\end{example}

Let us construct instances exactly as in Example~\ref{ex3} except that we use $m=n=200$ and $r=1,2,5,10,20$. 
Note that the advantage of using quantized low-rank matrices is that we know there there exists a solution with error smaller than 1/2. 
We stop the execution of Algorithm~\ref{altopt} only when 1000 iterations are performed or when the relative error between two iterates is smaller than $10^{-6}$, that is, 
when $e_{t}-e_{t+1} \leq 10^{-6} ||M_q||_{\infty}$, where $e_t$ is the error at iteration $t$. 
Table~\ref{tableQantG} provides 
the smallest, average and largest value of the error of Algorithm~\ref{altopt} (second column), the number of solutions found with error smaller than 0.5 (third column), 
the smallest, average and largest number of iterations needed to converge (fourth column), 
the average computational time (fifth column), 
and the smallest, average and largest value of the $\ell_{\infty}$ error for $\ell_2$ LRA which we used as an initialization for Algorithm~\ref{altopt} (last column). 
\begin{center}
 \begin{table}[h!]
 \begin{center}
\caption{ 
Results of Algorithm~\ref{altopt} on Gaussian random instances of the recovery of quantized low-rank matrices. 
\label{tableQantG}
}
 \begin{tabular}{|c||c|c|c|c||c|}
 \hline 
$r$ & Error of Alg.~\ref{altopt} & \# runs & \# it. of Alg.~\ref{altopt} & Average  & Error of $\ell_2$ LRA (init.) \\  
 & min , mean , max & error $\leq 0.5$ & min , mean , max & time (s.) & min , mean , max  \\ \hline 
 1   &   0.50 , 0.50 , 0.50  &   100/100   &   14 , \;49 , 110 &    0.67    &  0.92 , 0.96 , 0.97 \\ 
 2   &   0.53 , 0.69 , 0.87  &     0/100   &    \;3 ,  \;\;5 , \;36 &    0.16    &  0.82 , 0.94 , 0.97  \\ 
 5   &   0.52 , 0.54 , 0.62  &     0/100   &    \;3 , \;19 , 127 &    1.47    &  0.65 , 0.70 , 0.87  \\ 
10   &   0.50 , 0.52 , 0.56  &     0/100   &    \;5 , \;90 , 265 &    15.19    &  0.70 , 0.75 , 0.85  \\ 
20   &   0.48 , 0.49 , 0.53  &    93/100   &   14 , 175 , 333 &    58.69    &  0.74 , 0.81 , 0.91  \\  \hline  
	\end{tabular}
 \end{center}
 \end{table}
 \end{center}
We observe the following: 
\begin{itemize}

\item In all cases, Algorithm~\ref{altopt} is able to significantly improve the initial solution computed with $\ell_2$ LRA (second column vs.\@ last column of Table~\ref{tableQantG}). 

\item In many cases, Algorithm~\ref{altopt} converges in a relative few number of iterations (sometimes in 3 iterations); see the fourth column of Table~\ref{tableQantG}. We believe the reason is that the objective function landscape is rather peaky hence it is more likely for the algorithm to terminate rapidly.  

\item For $r=1$, Algorithm~\ref{altopt} is always able to recover a solution with error smaller than 1/2. 
This is not surprising since the problem is not difficult: in fact, the graph $G_b(M_q,0.5)$ contains a single connected component since most entries of $|M_q|$ are larger than 0.5 (see Theorem~\ref{theoremNP}). 
We observed in practice that, in this case, Algorithm~\ref{altopt} is always able to converge to an optimal solution. Optimality can be verified by checking that the answer to D-$\ell_\infty$-R1A($M_q$,$f^*-\epsilon$) is NO, where $f^*$ is the objective function value of Algorithm~\ref{altopt} at convergence and $\epsilon$ is a small positive constant. (We have included this function in the code available online.) 

\item When $r > 1$, the problem becomes more difficult and Algorithm~\ref{altopt} is not able to identify a solution with error smaller than 1/2 (it is never able to do it for $r=2,5,10$). This leads to an interesting question: 
is $\ell_{\infty}$ LRA hard for $r > 1$ (say $r=2$) even when $G_b(M_q,0.5)$ contains a unique connected component (or when $M \geq 0$)? 
Table~\ref{tableQantG} suggests that it is the case. At least, we observed that there are many local minima: for example, for $r=2$, using 100 random initializations does not allow in general to obtain a solution with error smaller than 1/2, while the solution $M$ cannot be improved; see Table~\ref{tableQantGv2}.  
Note also that the solutions obtained with random initializations have error significantly larger than with $\ell_2$-LRA initialization. 

\item Surprisingly, when $r = 20$, Algorithm~\ref{altopt} is able in most cases to recover a solution with error smaller than 1/2. We believe the reason is that the number of degrees of freedom is large hence the optimal solution has error smaller than 0.5. 
This is confirmed by the results in Table~\ref{tableQantGv2} where we have initialized Algorithm~\ref{altopt} using the original rank-$r$ matrix $M$ (hence the initial error is close to 0.5). We see that, quite naturally, it is able to identify better solutions than when initialized with the solution of $\ell_2$ LRA.  
\begin{center}
 \begin{table}[h!]
 \begin{center}
\caption{ 
Results for the BCD Algorithm~\ref{altopt} initialized using the solution to $\ell_2$ LRA of matrix $M_q$ (exactly as for the second and last row of Table~\ref{tableQantG}) and with the solution to $\ell_2$ LRA of $M = UV$, where each entry of $U$ and $V$ is generated using the Gaussian distribution $N(0,1)$. 
\label{tableQantGv2}
}
 \begin{tabular}{|c||c|c|c|c|}
 \hline 
 & Error of BCD & \# runs  & \# it. of BCD & Average  \\  
    & min , mean , max & error $\leq 0.5$ & min , mean , max &  time (s.)   \\ \hline 
 $r=2$, init.\@ $\ell_2$   &   0.52 , 0.67 , 0.88  &     0/100   &    3 ,   6 ,  35 &    0.19    \\ 
$r=2$,  init.\@ $M$        &   0.50 , 0.50 , 0.50  &   100/100   &    2 ,  10 ,  51 &    0.31 \\ \hline 
 $r=20$, init.\@ $\ell_2$  &   0.48 , 0.49 , 0.51  &    96/100   &   71 , 178 , 261 &    56.12    \\
 $r=20$, init.\@ $M$       &    0.46 , 0.46 , 0.46  &   100/100   &   66 ,  96 , 135 &    30.60    \\ \hline 
	\end{tabular} 
 \end{center}
 \end{table}
 \end{center} 

\end{itemize}

\section{Conclusion} 

In this paper, we have analyzed the component-wise $\ell_{\infty}$ low-rank matrix approximation problem. We proved that, even in the rank-one case, the decision version of this problem is NP-complete using a reduction from `NOT ALL EQUAL 3SAT' (Theorem~\ref{thrnpcom}).  
However, in the rank-one case when $M \geq 0$ or when $G_b(M,k)$ contains a few number of connected components, 
the problem can be solved in polynomial time (Theorem~\ref{theoremNP}). 
We then described a simple block coordinate descent method and applied it for the recovery of quantized low-rank matrices. 
We observed that, as expected, the algorithm is able to recover an optimal solution when $r=1$. 
However, as soon as $r > 1$, the problem becomes more difficult even when $G_b(M,k)$ contains a single connected component.

\section*{Acknowledgment} 

We are grateful to Laurent Jacques for pointing us out to the quantized low-rank matrix problem and for insightful discussions. Part of this paper was written at the University of Mons during a visit of the second author who is grateful for the invitation and for being introduced to this problem.

\small

\bibliographystyle{siam}
\bibliography{LRAlinf}

\end{document}